\pdfoutput=1
\documentclass[envcountsame]{llncs}

\usepackage{amssymb}
\usepackage{amsmath}
\usepackage{graphicx}
\usepackage{url}
\usepackage{aliascnt}
\usepackage[disable]{todonotes}
\usepackage{tikz}
\usetikzlibrary{calc}
\usetikzlibrary{decorations.markings}
\usetikzlibrary{calc}
\usetikzlibrary{patterns}

\usepackage{caption}
\usepackage{subcaption}
\usepackage{xspace}
\usepackage[toc,page]{appendix}

\usepackage{color}

\usepackage[ruled,vlined,linesnumbered]{algorithm2e}
\usepackage{hyperref}

\newcommand{\lb}{\mathop{\mathrm{l\hbox to 0pt{\kern .05em $\overline{\vbox to 0.9ex{\hbox to .5em {~}}}$\hss}n}}\nolimits}
\newcommand{\lgb}{\mathop{\mathrm{l\hbox to 0pt{\kern .05em $\overline{\vbox to 1.1ex{\hbox to .85em {~}}}$\hss}og}}\nolimits}
\usepackage{xparse}

\NewDocumentCommand{\xnewtheorem}{m o m}
 {%
  \IfNoValueTF{#2}
   {\newtheorem{#1}{#3}}
   {%
    \newaliascnt{#1}{#2}%
    \newtheorem{#1}[#1]{#3}%
    \aliascntresetthe{#1}%
    \expandafter\newcommand\csname #1autorefname\endcsname{#3}%
   }%
 }

\long\def\printmath #1{\ifmmode#1 \else$\printmath{#1}$\expandafter\xspace\fi}

\long\def\newproblem#1#2#3{
 \expandafter	\def \csname #1\endcsname {#2\xspace}
 \expandafter	\def \csname #1long\endcsname {#3\xspace}
 \expandafter	\def \csname #1math\endcsname {\mbox{#2} }
	}

\def\bO#1{\printmath{\mathcal{O}\left(#1\right)}}

\def\bOm#1{\printmath{\Omega \left(#1\right)}}
 
\long\def\newbasevar#1#2 {
	\expandafter\def\csname base#1math\endcsname##1##2##3##4{{^{##3}_{##4}{#2}_{##1}^{##2}}}
}
\long\def\newvarfour#1#2{
	\newbasevar{#1}{#2}
	\expandafter\def\csname #1\endcsname##1##2##3##4{\printmath{\csname base#1math\endcsname{##1}{##2}{##3}{##4}}}
}

\long\def\newvarthree#1#2{
	\newbasevar{#1}{#2}
	\expandafter\def\csname #1\endcsname##1##2##3{\printmath{\csname base#1math\endcsname{##1}{##2}{##3}{}}}
}

\long\def\newvartwo#1#2{
	\newbasevar{#1}{#2}
	\expandafter\def\csname #1\endcsname##1##2{\printmath{\csname base#1math\endcsname{##1}{##2}{}{}}}
}

\long\def\newvarone#1#2{
	\newbasevar{#1}{#2}
	\expandafter\def\csname #1\endcsname##1{\printmath{\csname base#1math\endcsname{##1}{}{}{}}}
}

\long\def\newvar#1#2{
	\newbasevar{#1}{#2}
	\expandafter\def\csname #1\endcsname{\printmath{\csname base#1math\endcsname{}{}{}{}}}
}

\long\def\newoperator#1#2#3{
	\expandafter\def\csname #1\endcsname ##1##2{\printmath{##1#2##2#3}}
}

\long\def\varoperator#1#2#3{
	\expandafter\def\csname #1\endcsname ##1{\printmath{\csname #2\endcsname{#3}{##1}}}
}

\long\def\newbinaryrelation#1#2{
	\expandafter\def\csname #1\endcsname ##1##2{\printmath{##1#2##2}}
}

\long\def\newntupel#1#2{
	\expandafter\def\csname #1\endcsname ##1##2{\printmath{\left(##1,##2\right)}}
}

\long\def\newsetprinter#1#2{
	\expandafter\def\csname #1\endcsname ##1##2{\printmath{\left\{##1#2##2\right\}}}
}

\usepackage{etoolbox}
\usepackage{xstring}

\long\def\makeparamtovar#1#2#3#4#5#6#7#8#9{
	\expandafter\def\csname#2\endcsname{}
	\ifnumcomp{#1}{=}{1}
		{\expandafter\def\csname#2\endcsname{#3}}
		{}
	\ifnumcomp{#1}{=}{2}
		{\expandafter\def\csname#2\endcsname{#4}}
		{}
	\ifnumcomp{#1}{=}{3}
		{\expandafter\def\csname#2\endcsname{#5}}
		{}
	\ifnumcomp{#1}{=}{4}
		{\expandafter\def\csname#2\endcsname{#6}}
		{}
	\ifnumcomp{#1}{=}{5}
		{\expandafter\def\csname#2\endcsname{#7}}
		{}
	\ifnumcomp{#1}{=}{6}
		{\expandafter\def\csname#2\endcsname{#8}}
		{}
	\ifnumcomp{#1}{=}{7}
		{\expandafter\def\csname#2\endcsname{#9}}
		{}
}

\long\def\createvar#1#2#3{
	\StrLen{#3}[\totalargc]
	\expandafter\edef\csname #1totalargc\endcsname{\totalargc} 
	\StrLen{#3}[\argc]
	\StrRight{#3}{\argc}[\argv]

	\expandafter\edef\csname #1saifunc\endcsname{-1}
	\expandafter\edef\csname #1saiindex\endcsname{-1}
	\expandafter\edef\csname #1saipraeindex\endcsname{-1}
	\expandafter\edef\csname #1saiexp\endcsname{-1}
	\expandafter\edef\csname #1saipraeexp\endcsname{-1}

	\newcounter{cnt#1}
	\setcounter{cnt#1}{1}
	\whileboolexpr{ not (test { \ifnumequal{\argc}{0} } )
	}{
		\StrLeft{\argv}{1}[\param]	
		\StrGobbleLeft{\argv}{1}[\argv]	
		\IfStrEqCase{\param}{
			{f}{
				\expandafter\edef\csname #1saifunc\endcsname{\arabic{cnt#1}}	
			}
			{i}{ 
				\expandafter\edef\csname #1saiindex\endcsname{\arabic{cnt#1}}
			}
			{e}{ 
				\expandafter\edef\csname #1saiexp\endcsname{\arabic{cnt#1}}
			}
			{u}{ 
				\expandafter\edef\csname #1saipraeindex\endcsname{\arabic{cnt#1}}
			}
			{w}{ 
				\expandafter\edef\csname #1saipraeexp\endcsname{\arabic{cnt#1}}
			}
		}
		[bulub] 

		\StrLen{\argv}[\argc]
		\stepcounter{cnt#1}
	}
	\long\expandafter\gdef\csname base#1sort\endcsname##1##2##3##4##5{
		\expandafter\def\csname #1localfirst\endcsname{\ifnumcomp{\csname #1totalargc\endcsname}{>}{0}{##1}{}}
		\expandafter\def\csname #1localsec\endcsname{\ifnumcomp{\csname #1totalargc\endcsname}{>}{1}{##2}{}}
		\expandafter\def\csname #1localthird\endcsname{\ifnumcomp{\csname #1totalargc\endcsname}{>}{2}{##3}{}}
		\expandafter\def\csname #1localfourth\endcsname{\ifnumcomp{\csname #1totalargc\endcsname}{>}{3}{##4}{}}
		\expandafter\def\csname #1localfifth\endcsname{\ifnumcomp{\csname #1totalargc\endcsname}{>}{4}{##5}{}}
		\makeparamtovar{\csname #1saifunc\endcsname}{#1tmpfunc}{\csname #1localfirst\endcsname}{\csname #1localsec\endcsname}{\csname #1localthird\endcsname}{\csname #1localfourth\endcsname}{\csname #1localfifth\endcsname}{}{}
		\makeparamtovar{\csname #1saiindex\endcsname}{#1tmpindex}{\csname #1localfirst\endcsname}{\csname #1localsec\endcsname}{\csname #1localthird\endcsname}{\csname #1localfourth\endcsname}{\csname #1localfifth\endcsname}{}{}
		\makeparamtovar{\csname #1saipraeindex\endcsname}{#1tmppraeindex}{\csname #1localfirst\endcsname}{\csname #1localsec\endcsname}{\csname #1localthird\endcsname}{\csname #1localfourth\endcsname}{\csname #1localfifth\endcsname}{}{}
		\makeparamtovar{\csname #1saiexp\endcsname}{#1tmpexp}{\csname #1localfirst\endcsname}{\csname #1localsec\endcsname}{\csname #1localthird\endcsname}{\csname #1localfourth\endcsname}{\csname #1localfifth\endcsname}{}{}
		\makeparamtovar{\csname #1saipraeexp\endcsname}{#1tmppraeexp}{\csname #1localfirst\endcsname}{\csname #1localsec\endcsname}{\csname #1localthird\endcsname}{\csname #1localfourth\endcsname}{\csname #1localfifth\endcsname}{}{}
		\sbox0{$\csname #1tmppraeindex\endcsname$}%
		\ifdim\wd0=0pt
			{}
		\else
			_{\csname #1tmppraeindex\endcsname}
		\fi
		\sbox1{$\csname #1tmppraeexp\endcsname$}%
		\ifdim\wd1=0pt
			{}
		\else
			^{\csname #1tmppraeexp\endcsname}
		\fi
		{#2}
		\sbox2{$\csname #1tmpindex\endcsname$}%
		\ifdim\wd2=0pt
			{}
		\else
			_{\csname #1tmpindex\endcsname}
		\fi
		\sbox3{$\csname #1tmpexp\endcsname$}%
		\ifdim\wd3=0pt
			{}
		\else
			^{\csname #1tmpexp\endcsname}
		\fi
		\sbox4{$\csname #1tmpfunc\endcsname$}%
		\ifdim\wd4=0pt
			{}
		\else
			{\left(\csname #1tmpfunc\endcsname\right)}
		\fi
	}

	\ifnumcomp{\csname #1totalargc\endcsname}{=}{0}
		{\long\expandafter\gdef\csname #1\endcsname{\printmath{\csname base#1sort\endcsname{}{}{}{}{}}}}
		{}
	\ifnumcomp{\csname #1totalargc\endcsname}{=}{1}
		{\long\expandafter\gdef\csname #1\endcsname##1{\printmath{\csname base#1sort\endcsname{##1}{}{}{}{}}}}
		{}
	\ifnumcomp{\csname #1totalargc\endcsname}{=}{2}
		{\long\expandafter\gdef\csname #1\endcsname##1##2{\printmath{\csname base#1sort\endcsname{##1}{##2}{}{}{}}}}
		{}
	\ifnumcomp{\csname #1totalargc\endcsname}{=}{3}
		{\long\expandafter\gdef\csname #1\endcsname##1##2##3{\printmath{\csname base#1sort\endcsname{##1}{##2}{##3}{}{}}}}
		{}
	\ifnumcomp{\csname #1totalargc\endcsname}{=}{4}
		{\long\expandafter\gdef\csname #1\endcsname##1##2##3##4{\printmath{\csname base#1sort\endcsname{##1}{##2}{##3}{##4}{}}}}
		{}
	\ifnumcomp{\csname #1totalargc\endcsname}{=}{5}
		{\long\expandafter\gdef\csname #1\endcsname##1##2##3##4##5{\printmath{\csname base#1sort\endcsname{##1}{##2}{##3}{##4}{##5}}}}
		{}
}

\tikzstyle{key} = [circle,draw=black!100]
\tikzstyle{node} = [fill=black,circle,inner sep=1pt]

\colorlet{color_1}{red}
\colorlet{color_2}{blue}
\colorlet{color_3}{green}
\colorlet{color_4}{yellow}
\colorlet{base}{black}
\colorlet{bg}{white}

\colorlet{beamer_base_color}{blue}
\colorlet{beamer_primary_color}{beamer_base_color}
\colorlet{beamer_secondary_color}{beamer_base_color!50!black}
\colorlet{beamer_tertiary_color}{beamer_base_color!30!black}
\colorlet{beamer_quaternary_color}{black}

\begin{document} 

\pagestyle{plain}
\title{On the Complexity of List Ranking in the Parallel External Memory Model}

\institute{
  Institute for Theoretical Computer Science, ETH Z{\"u}rich, Switzerland \email{\{rjacob,lieberto\}@inf.ethz.ch}
  \and
  Department of Information and Computer Sciences, University of Hawaii, USA \email{nodari.sitchinava@hawaii.edu}
}
\author{
  Riko Jacob \inst{1}
  \and
  Tobias Lieber \inst{1} 
  \and
  Nodari Sitchinava \inst{2}
}
\date{\today}

\maketitle

\createvar{memorySize}{M}{}
\createvar{blockSize}{B}{}

\createvar{inputSize}{N}{}
\begin{abstract}
\footnotetext[3]{This paper is also published in the proceeding of MFCS 2014, excluding the appendix. The final publication is available at \texttt{link.springer.com}}
We study the problem of {\em list ranking} in the parallel external memory (PEM) model. 
We observe an interesting dual nature for the hardness of the problem due to limited information exchange among the processors about the structure of the list, on the one hand, and its close relationship to the problem of permuting data, which is known to be hard for the external memory models, on the other hand. 

By carefully defining the power of the computational model, we prove a permuting lower bound in the PEM model.
Furthermore, we present a stronger $\Omega(\log^2 \inputSize)$ lower bound for a special variant of the problem and for a specific range of the model parameters, which takes us a step closer toward proving a non-trivial lower bound for the list ranking problem in the bulk-synchronous parallel (BSP) and MapReduce models.
Finally, we also present an algorithm that is tight for a larger range of parameters of the model than in prior work.

%
%


\end{abstract}

\createvar{graph}{G}{}
\createvar{vertexSet}{V}{}
\createvar{edgeSet}{E}{}
\createvar{graphOf}{\graph}{i}

\createvar{sort}{\textnormal{sort}}{f}
\createvar{sortp}{\sort{}_{P}}{f}

\section{Introduction}
\label{sectionIntroduction}
\createvar{permp}{\textnormal{perm}_P}{f}
Analysis of massive graphs representing social networks using distributed programming models, such as MapReduce and Hadoop, has renewed interests in distributed graph algorithms. 
In the classical RAM model, depth-first search traversal of the graph is the building block for many graph analysis solutions. 
However, no efficient depth-first search traversal is known in the parallel/distributed setting. 
Instead, list ranking serves as such a building block for parallel solutions to many problems on graphs. 

The \emph{list ranking} problem is defined as follows: 
given a linked list compute for each node the length of the path to the end of the list.  
In the classic RAM model, the list can be ranked in linear time by traversing the list.
However, in the PRAM model (the parallel analog of the RAM model) it took almost a decade from the first solution by Wyllie~\cite{1979Wyllie} till it was solved optimally~\cite{1988AndersonMillerDetLR}. 

The problem is even more intriguing in the models that study block-wise access to memory. 
For example, in the \emph{external memory} (EM) model of Aggarwal and Vitter~\cite{1988AggarwalVitterEM} list ranking is closely related to the problem of permuting data in an array.
The EM model studies the {\em input/output (I/O) complexity} -- the number of transfers an algorithm has to perform between a disk that contains the input and a fast internal memory of size $\memorySize$.
Each transfer is performed in blocks of $\blockSize$ contiguous elements.  
In this model, permuting and, consequently, list ranking require I/O complexity which is closely related to sorting~\cite{1995ChiangExternalMemoryGraphAlgorithms}, rather than the linear complexity required in the RAM model.

In the distributed models, such as bulk-synchronous parallel (BSP)~\cite{1990ValiantBSP} or MapReduce~\cite{2008DeanMapReduce} models, the hardness of the list ranking problem is more poorly understood. These models consist of $P$ processors, each with a private memory of size $\memorySize$. With no other data storage, typically $\memorySize = \Theta(N/P)$. The data is exchanged among the processors during the \emph{communication rounds} and the number of such rounds defines the complexity metric of these models. 

One established modeling of today's commercial data centers running MapReduce is to assume $P = \Theta(N^\epsilon)$ and $M = \Theta(N^{1-\epsilon})$ for a constant $0< \epsilon < 1$. 
Since network bandwidth is usually the limiting factor of these models, $O(\log_M P) = O(1)$ communication rounds is the ultimate goal of computation on such models~\cite{2010KarloffMapReduce}.
Indeed, if each processor is allowed to send up to $M = N/P$ items to any subset of processors, permuting of the input can be implemented in a single round, while sorting takes $O(\log_M P) = O(1)$ rounds~\cite{1999GoodrichBSPSort,2011GoodrichMapReduceSort}.
On the other hand, the best known solution for list ranking is via the simulation results of Karloff et al.~\cite{2010KarloffMapReduce} by simulating the $O(\log N)$ time PRAM algorithm~\cite{1988AndersonMillerDetLR}, yielding $O(\log P)$ rounds, which is strictly worse than both sorting and permuting.  
Up to now, no non-trivial lower bounds (i.e. stronger than $\bOm{\log_M P} = \bOm{1}$) are known in the BSP and MapReduce models.  

In this paper we study lower bounds for the list ranking problem in the {\em parallel external memory (PEM)} model. 
%
%
The PEM model was introduced by Arge et al.~\cite{2008ArgeEtAlFundamentalPEM} as a parallel extension of the EM model to capture the hierarchical memory organization of modern multicore processors. 
The model consists of~$P$ processing units, each containing a private cache of size $\memorySize$, and a shared (external) memory of conceptually unlimited size (see Figure~\ref{figurePEMmodel}). 
The data is concurrently transferred between the processors' caches and shared memory in blocks of size $\blockSize$. 
The model measures the {\em parallel I/O complexity} -- the number of parallel block transfers. 
\begin{figure}[tb]
\centering
\begin{subfigure}{.3\textwidth}
\includegraphics[height=95pt]{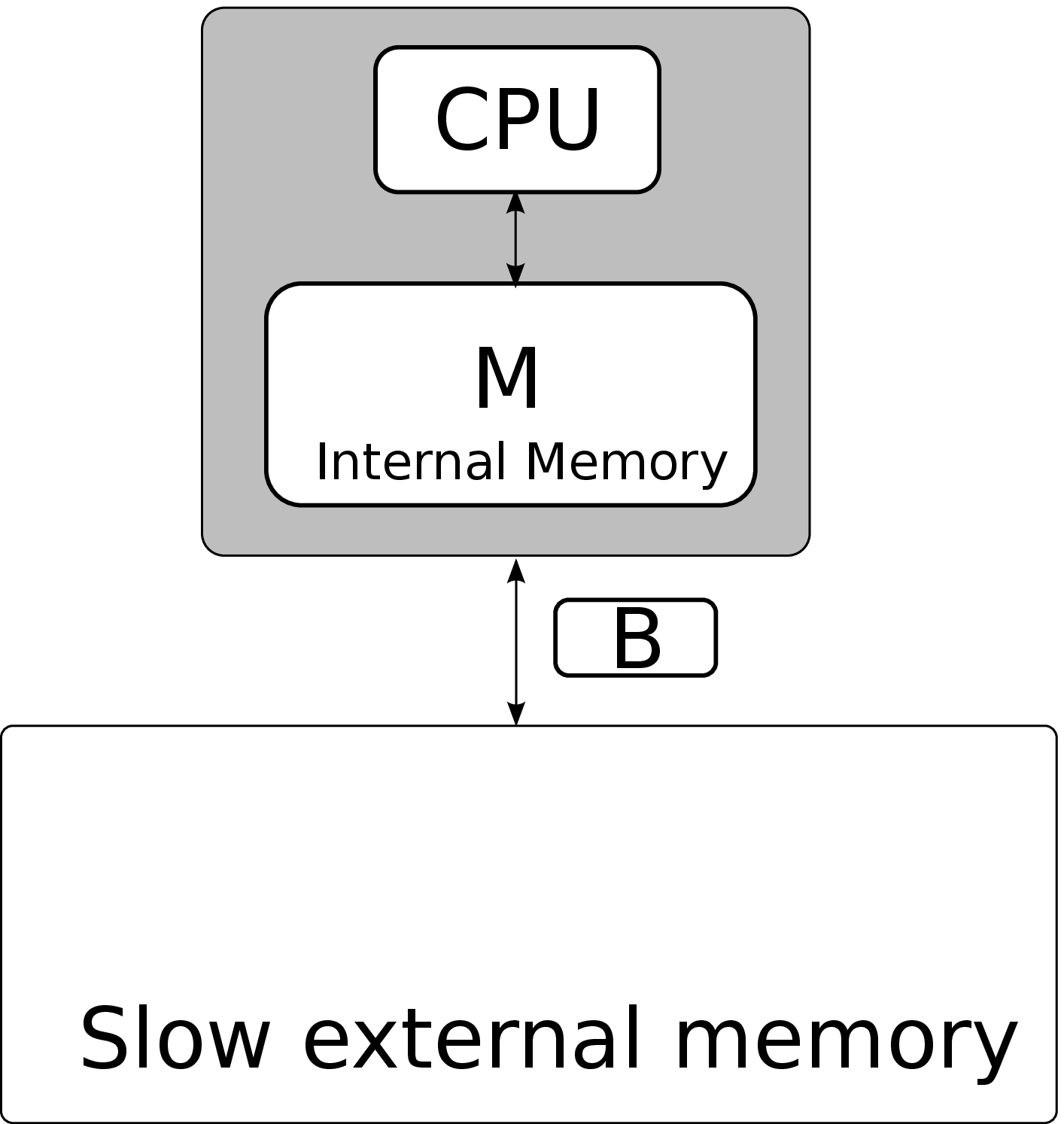}
\caption{The EM Model}
\label{figureEMmodel}
\end{subfigure}
\begin{subfigure}{.69\textwidth}
\includegraphics[height=100pt]{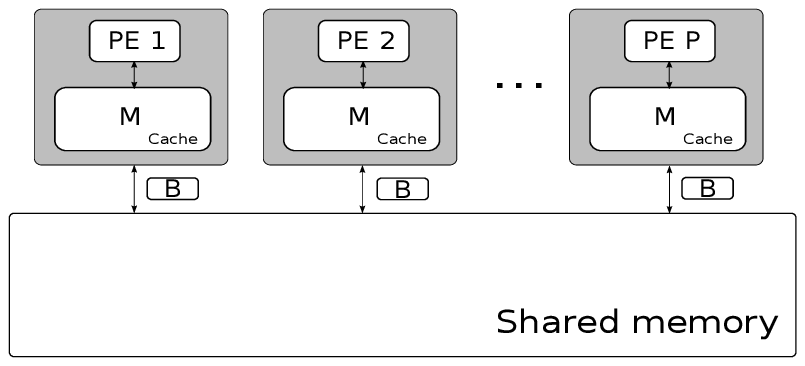}
\caption{The PEM Model}
\label{figurePEMmodel}
\end{subfigure}
\caption{The sequential and parallel external memory models.}
\end{figure}
From the discussion above, it appears that the hardness of list ranking stems from two factors: (1) limited speed of discovery of the structure of the linked list due to limited information flow among the processors,
and (2) close relationship of list ranking to the problem of permuting data. 
While only one of these challenges is captured by the distributed models or the sequential EM model, both of them are exhibited in the PEM model: the first one is captured by the distributed nature of the private caches of the model, and the second one has been shown by Greiner~\cite{2012ThesisGero} who proves that permuting data in the PEM model takes asymptotically $\permp{\inputSize,M,B} = \min\left\{\frac{{\inputSize}}{P},\frac{\inputSize}{PB}\lgb_d\frac{\inputSize}{B}\right\}$ parallel I/Os, where $d=\max\left\{2,\min\left\{\frac{M}{B},\frac{\inputSize}{PB}\right\}\right\}$ and $\lgb(x) = \max\{1, \log(x)\}$.

Part of the challenge of proving lower bounds (in any model) is restricting the model enough to be able to prove non-trivial bounds, while identifying the features of the model that emphasize the hardness of a particular problem. 

An example of such restriction in the external memory models (both sequential and parallel) is the so-called {\em indivisibility assumption}~\cite{1988AggarwalVitterEM}. 
The assumption states that each item is processed as a whole, and no information can be obtained from a part of the input, for example, by combining  several items into one. 
To our knowledge, without the indivisibility assumption, it is not clear how to prove lower bounds in the PEM model exceeding the information-theoretic lower bounds of \bOm{\log P} parallel I/Os \cite{1990KarpRamachandranPRAMListRankLB,2008ArgeEtAlFundamentalPEM}.  

\subsection{Our Contributions}
\label{sectionContributions}

\createvar{movePEM}{\text{atomic PEM}}{}
\createvar{bridgingPEM}{\text{edge-contracting PEM}}{} 
\createvar{semigroupPEM}{\text{semigroup PEM}}{}
\createvar{fusePEM}{\text{interval PEM}}{}

In this paper, we address the precise formulation of the power of the PEM model for the list ranking problem.
In \autoref{sectionModelling} we present the \movePEM model which formalizes the indivisibility assumption in the PEM model.
It can be viewed as the parallel analog of the model for proving permuting lower bounds in the sequential EM model~\cite{1988AggarwalVitterEM}. 
We extend this basic model by allowing an algorithm to perform operations on the atoms that create new atoms. 
While we always keep the indivisibility of the atoms, the precise operations and the information the algorithm has about the content of the atom varies.\todo{probably not needed?}

In the sequential EM model, Chiang et al.~\cite{1995ChiangExternalMemoryGraphAlgorithms} sketch a lower bound for list ranking via a reduction to the {\em proximate neighbors problem}.
However, since there is no equivalent to Brent's scheduling principle~\cite{1984VishkinRandomizedLR} in the PEM model~\cite{2012ThesisGero}, the lower bound does not generalize to the PEM model.

Therefore, in \autoref{sectionCountingLowerBounds} we derive a lower bound of $\Omega(\permp{\inputSize, \memorySize, \blockSize})$ parallel I/Os for the proximate neighbor problem, and two problems which are related to the list ranking problem. Our lower bounds hold for both deterministic and randomized algorithms.
In the process we provide an alternative proof for the proximate neighbor lower bound in the sequential external memory model, matching the result of Chiang et al.~\cite{1995ChiangExternalMemoryGraphAlgorithms}\todo{needed?}.
Those lower bounds essentially exploit the fact that the same problem can be represented as input in many different layouts.

The discussion in \autoref{sectionIntroduction} about the dual nature of hardness of the list ranking problem might hint at the fact that the result above is only part of the picture and a stronger lower bound might be achievable. 
Part of the challenge in proving a stronger lower bound lies in the difficulty of combining the indivisibility assumption with the restrictions on how the structure of the linked list is shared among the processors without giving the model too much power, thus, making the solutions trivial. 
We address this challenge by defining the \emph{\fusePEM} model and defining the {\em guided interval fusion (GIF)} problem (\autoref{sectionIPNLB}).
We prove that GIF requires $\Omega(\log^2 \inputSize)$ parallel I/Os in the \fusePEM. Our lower bound for GIF in the PEM model implies a $\Omega(\log \inputSize)$ lower bound for the number of rounds for GIF in the distributed models when $P = \Theta(M) = \Theta(\sqrt{\inputSize})$. 

GIF captures the way how all currently known algorithms use information to solve list ranking in all parallel/distributed models.
Therefore, if this lower bound could be broken for the list ranking problem, it will require completely new algorithmic techniques.
Thus, our result brings us a step closer to proving the unconditional $\Omega(\log P)$ lower bound in the BSP and MapReduce models.

Finally, in \autoref{sectionUpperBounds} we improve the PEM list ranking algorithm of Arge et al.~\cite{2010ArgeEtAlGraphAlgoPEM} to work efficiently for a larger range of parameters in the PEM model.

\section{Modeling}
\label{sectionModelling}

We extend the description of the PEM model given in \autoref{sectionIntroduction} to define the PEM model more precisely. 
Initially, the data resides in the shared memory. 
To process any element, it must be present in the corresponding processor's cache.
The shared memory is partitioned into blocks of $\blockSize$ contiguous elements and the transfer of data between the shared memory and caches is performed by transferring these blocks as units.
Each transfer, an {\em input-output operation}, or simply {\em I/O}, can transfer one block of $\blockSize$ items between the main memory and each processors' cache. 
Thus, up to $P$ blocks can be transferred in each \emph{parallel I/O} operation. 
The complexity measure of a PEM algorithm is the number of parallel I/Os that the algorithm performs.  

Similar to the PRAM model, there are several policies in the PEM model, for handling simultaneous accesses by multiple processors to a block in shared memory. 
In this paper we consider the CREW PEM model, using a block wise concurrent read, exclusive write policy.


In order to prove lower bounds we make the definition of the model more precise by stating what an algorithm is able to do in each step.
In particular, we assume that each element of the input is an indivisible unit of data, an {\em atom}, which consumes one memory cell in the cache or shared memory.
Such atoms come into existence either as input atoms, or by an operation, as defined later, on two atoms.
A program or an algorithm has limited knowledge about the content of an atom. 
In this paper an atom does not provide any information. 
Furthermore, the \emph{\movePEM} is limited to the following operations: 
an I/O operation reads or writes one block of up to~$B$ atoms in the shared memory, and atoms can be copied or deleted.
Formal definitions of similar PEM machines can be found in \cite{2008ArgeEtAlFundamentalPEM,2012ThesisGero}.

For providing lower bounds for different problems, the concept of the \linebreak\movePEM is extended in later sections.

In the following, we distinguish between {algorithms} and {programs} in the following way:
In an \emph{algorithm} the control flow might depend on the input, i.e., there are conditional statements (and therefore loops).
In contrast, a \emph{program} has no conditional statements and is a sequence of valid instructions for a PEM model, independent of the input (atoms). 
For a given instance of a computational task, a program can be seen as an instantiation of an algorithm to which all results of conditional statements and index computations are presented beforehand.
Note, that in the problems considered in \autoref{sectionCountingLowerBounds} the copying and the deletion operation of the \movePEM do not help at all, since a program can be stripped down to operations which operate on atoms which influence the final result. 

%

\section{Counting Lower Bounds to the List Ranking Problem}
\label{sectionCountingLowerBounds}
In this section we prove the lower bound for the list ranking problem by showing the lower bound to the proximate neighbors problem~\cite{1995ChiangExternalMemoryGraphAlgorithms} and reducing it to the problems of semigroup evaluation, edge contraction and, finally, list ranking. 

\subsection{Proximate Neighbors Problem in PEM}

\createvar{labelfunction}{\lambda}{f}
\createvar{labeledAtoms}{L}{}

\begin{definition}[\cite{1995ChiangExternalMemoryGraphAlgorithms}]
  A \emph{block permutation} describes the content of the shared memory of a PEM configuration as a set of at most~$B$ atoms for each block.
\end{definition}

\begin{definition}
  An instance of the \emph{proximate neighbors problem} of size \inputSize consists of atoms~$x_i$ for $i\in[\inputSize]$.
  All atoms are labeled by a labeling function \mbox{\labelfunction{}: $[\inputSize] \mapsto [\frac{\inputSize}{2}]$} with $|\labelfunction{}^{-1}(i)|=2$.
  An output block permutation solves the problem if for every~$i\in[\frac{\inputSize}{2}]$ the two neighboring atoms $\labelfunction{}^{-1}(i)$ are stored in the same block. 
  The blocks in such an output may contain less than~$B$ atoms.
\end{definition}


\createvar{permutationProblem}{A}{}
\begin{lemma}
\label{theoremGeneralizedPermutingTasksLbPEM}
Let \permutationProblem be a computational problem of size \inputSize for which an algorithm has to be capable of generating at least $\left(\frac{\inputSize}{eB}\right)^{cN}$ block permutations, for a constant~\mbox{$c >0$}. 
Then in the CREW \movePEM model with $P\leq \frac{\inputSize}{B}$ processors, at least half of the input instances of \permutationProblem require 
\bOm{\permp{\inputSize,M,B}} parallel I/Os. 
\end{lemma}

\begin{proof}
Straightforward generalization of the proof to Theorem $2.7$ in~\cite{2012ThesisGero}.
\qed
\end{proof}

\begin{theorem}
\label{lemmaProximateNeighborsLbPEM} 
At least half of the instances of the proximate neighbors problem of size $N$ require \bOm{\permp{\inputSize,\memorySize,\blockSize}} parallel I/Os in the CREW atomic PEM model with $P < N/B$ processors. 
\end{theorem}
\begin{proof}
In \autoref{appendixProximateNeighbors} we show that any output block permutation of the proximate neighbors problem solves at least~$\frac{({\inputSize}/2)!}{(B/2)^\frac{\inputSize}{2}}$ in block permutations.
The theorem follows from \autoref{theoremGeneralizedPermutingTasksLbPEM} and the observation that $\frac{({\inputSize}/2)!}{(B/2)^\frac{\inputSize}{2}} \geq \left(\frac{\inputSize}{eB}\right)^\frac{\inputSize}{2}$. \qed
\end{proof}
Note, that the bound holds even if a program has full access to the labeling function \labelfunction{} and thus is fully optimized for an input.
The origin of the complexity of the problem rather is permuting the atoms to the output block permutation that solves the problem.

\subsection{Semigroup Evaluation in the PEM Model}
\label{sectionSemigroupLB}
\createvar{semigroup}{S}{}
\createvar{semigroupSet}{S}{}
\createvar{semigroupOperation}{\cdot}{}
\createvar{permutation}{\pi}{f}


Consider the problem of evaluating a very simple type of expressions, namely that of a semigroup, in the PEM model.

\begin{definition}[Semigroup Evaluation]
\label{definitionSemigroupEvaluation} 
Let \semigroup be a semigroup with its associative binary operation $\semigroupOperation: \semigroup\times\semigroup \rightarrow \semigroup$.
The \emph{semigroup evaluation problem} is defined as evaluating the expression $\prod_{i=1}^{\inputSize}a_i$, with $ a_i=x_{\permutation{i}}$, for the array of input atoms $x_i\in\semigroup$ for $1\leq i\leq \inputSize$ and where $\permutation{}$ is a permutation over~$[\inputSize]$.
\end{definition}

To be able to solve the semigroup evaluation problem, algorithms must be able to apply the semigroup operation to atoms.
Thus, we extend the \movePEM model to the \emph{\semigroupPEM} model by the following additional operation: 
if two atoms $x$ and $y$ are in the cache of a processor, a new atom $z=x\semigroupOperation y$ can be created.

We say that a program is correct in the \semigroupPEM if it computes the correct result for any input and any semigroup.


\begin{theorem}
\label{theoremLowerBoundSemiGroup} 
At least one instance of the semigroup evaluation problem of size~\inputSize requires \bOm{\permp{\inputSize,\memorySize,\blockSize}} parallel I/Os in the CREW \semigroupPEM model with $P\leq\frac{\inputSize}{\blockSize}$ processors.
\end{theorem}

\createvar{iosSemigroupEvaluation}{\ell}{}
\createvar{hardInstance}{I}{}
\createvar{semigroupProblem}{\mathcal{S}}{}
\createvar{proximateProblem}{\mathcal{P}}{}
\createvar{programOf}{\program}{i}
\createvar{semigroupProgram}{\programOf{\semigroupProblem}}{}
\createvar{instructionSequence}{\mathcal{I}}{i}

\def\upP{{\proximateProblem{}}}
\def\upS{{\semigroupProblem{}}}

\createvar{labf}{\lambda}{}
\createvar{instPN}{I^\upP_{\labf}}{}
\createvar{progPN}{P^\upP_{\labf}}{}
\createvar{perm}{\pi}{f}
\createvar{instSE}{I^\upS_{\perm{}}}{}
\createvar{progSE}{P^\upS_{\perm{}}}{}
\createvar{MPN}{\memorySize_{\proximateProblem{}}}{}
\createvar{BPN}{\blockSize_{\proximateProblem{}}}{}
\createvar{MSE}{\memorySize_{\semigroupProblem{}}}{}
\createvar{BSE}{\blockSize_{\semigroupProblem{}}}{}
\createvar{tPN}{t_{\proximateProblem{}}(\inputSize,\MPN,\BPN)}{}
\createvar{tSE}{t_{\semigroupProblem{}}(\inputSize,\MSE,\BSE)}{}
\createvar{inputAtoms}{X}{}

\begin{proof}[Sketch]
Let~\instPN be an instance of the proximate neighbors problem over the input atoms $\inputAtoms=\{x_i|i\in[\inputSize]\}$ with its labeling function~\labf.
We consider an instance~\instSE of the semigroup evaluation problem over the semigroup on the set $\inputAtoms^2$ with the semigroup operation $(a,b)\semigroupOperation(c,d) = (a,d)$, where $a,b,c,d\in \inputAtoms$.
Furthermore, the instance~\instSE is defined over the input atoms $a_i=(x_i,x_i)$, with $1\leq i\leq \inputSize$.
The permutation \perm{} of~\instSE is one of the permutations such that for all $i\in\left[\frac{\inputSize}{2}\right]$, $\{\perm{2i-1},\perm{2i}\} = \labf{}^{-1}(i)$ holds.

Then, the key idea is to write for each application of the semigroup operation $(a,b)\semigroupOperation(c,d)$ in a program solving~\instSE, the pair $\{b,c\}$ as a result for \instPN to the output.
This would yield an efficient program for~\instPN, and therefore yields the lower bound by \autoref{lemmaProximateNeighborsLbPEM}.
We present the complete proof in \autoref{appendixSemigroup}.
\qed
\end{proof}

\subsection{Atomic Edge Contraction in the PEM Model} 
\label{sectionListRankingLB}


\begin{definition}
The input of the {\em atomic edge contraction problem} of size \inputSize consists of atoms $x_i$, $1\leq i \leq \inputSize$, which represent directed edges $e_i$ on a $(\inputSize+1)$-vertex path between vertices $s$ and $t$. 
Initially, the edges are located in arbitrary locations of the shared memory.
The instance is solved if an atom representing the edge $(s,t)$ is created and written to shared memory.
\end{definition}

To prove the lower bound for the atomic edge contraction problem, we extend the \movePEM with an additional operation: 
two atoms representing a pair of edges $(a,b)$ and $(b,c)$ can be removed and replaced by a new atom representing a new edge $(a,c)$.
We call the resulting model {\em \bridgingPEM}.



\begin{theorem} 
\label{lemmaAtomicEdgeContraction}
There is at least one instance of the atomic edge contraction problem of size~\inputSize which requires \bOm{\permp{\inputSize,\memorySize,\blockSize}} parallel I/Os in the CREW \bridgingPEM model with $P\leq\frac{\inputSize}{\blockSize}$ processors.
\end{theorem}

\createvar{atomPermutation}{\perm{}}{f}
\createvar{atomicEdgeContractionProblem}{\mathcal{E}}{}
\createvar{instanceAtomicListRanking}{I^{\atomicEdgeContractionProblem}}{}
\createvar{programAtomicListRanking}{P^{\atomicEdgeContractionProblem}}{}

\begin{proof}[Sketch]
An instance \instSE of the semigroup evaluation problem can be reduced to an instance \instanceAtomicListRanking of the atomic edge contraction problem, by defining the atom $x_{\perm{i}}$ of \instanceAtomicListRanking, initially stored at location $\perm{i}$, as \mbox{$e_{\perm{i}}=(\perm{i},\perm{i+1})$}, where \perm{} is the permutation of \instSE. 
The full proof is given in \autoref{appendixAtomicEdgeContraction}.
\end{proof}


\subsection{Randomization and Relation to the List Ranking Problem}
Observe that the expected number of parallel I/Os of a randomized algorithm for an instance is a convex combination of the number of parallel I/Os of programs. 
Combining this observation with the \bOm{\log P} lower bound of \cite{1990KarpRamachandranPRAMListRankLB,2008ArgeEtAlFundamentalPEM} mentioned in \autoref{sectionIntroduction} we obtain:

\begin{theorem}
For the proximate neighbors, semigroup evaluation, and the atomic edge contraction problems, there exists at least one instance that requires at least \bOm{\permp{\inputSize,\memorySize,\blockSize} + \log P} expected parallel I/Os by any randomized algorithm in the corresponding PEM model with $P\leq\frac{\inputSize}{\blockSize}$ processors. 
\end{theorem}


Although our \semigroupPEM and \bridgingPEM models might seem too restrictive at a first glance.
To the best of our knowledge all current parallel solutions to list ranking utilize pointer hopping, which can be reduced to atomic edge contraction and thus the lower bound applies. 




\section{The Guided Interval Fusion Problem (GIF)}

\label{sectionIPNLB}

\createvar{proximateNeighborsAlgorithm}{A_p}{}


In this section we prove for the GIF problem, which is very similar to the atomic edge contraction problem, a lower bound  of \bOm{\log^2 \inputSize} in the PEM model with parameters $P=\memorySize$ and $\blockSize=\memorySize /2$ for inputs of size $\inputSize = P\memorySize=2^x$ for some~\mbox{$x\in\mathbb{N}$}. 
In contrast to the atomic edge contraction problem, in the GIF problem an algorithm is not granted unlimited access to the permutation \permutation{}.

The chosen parameters of the PEM model complement the upper bounds of \autoref{sectionUpperBounds} at one specific point in the parameter range.
Note that with careful modifications the \bOm{\log^2 \inputSize} bound can even be proven for $\inputSize =\memorySize^{\frac{3}{2}+\varepsilon}$.

\createvar{gifInstance}{\mathcal{G}}{}
\createvar{neighborOf}{n}{f}

\begin{definition}
The \emph{\fusePEM} is an extension of the \movePEM:
Two atoms $x$ and $y$ representing closed intervals $I_x$ and $I_y$, located in one cache, can be \emph{fused} if $I_x\cap I_y \neq \emptyset$.
Fusing creates a new atom $z$ representing the interval $I_z=I_x\cup I_y$. 
We say $z$ is derived from $x$ if $z$ is the result of zero or more fusing operations starting from atom~$x$.
\end{definition}

\begin{definition}
The \emph{guided interval fusion problem} (GIF) is a game between an algorithm and an adversary, played on an \fusePEM.
The algorithm obtains a GIF instance \gifInstance in the first \inputSize cells of the shared memory, containing \inputSize uniquely named atoms $x_i$, $1\leq i \leq\inputSize$.
Each initial atom $x_i$ represents the (invisible to the algorithm ) closed interval $I_{x_i}=[k-1,k]$ for $k=\permutation{i}$ with $1\leq k \leq \inputSize$.

The permutation $\permutation{}$ is gradually revealed by the adversary in form of \emph{boundaries} $p=(i,j)$ meaning that the (initial) atoms $x_i$ and $x_j$ represent neighboring intervals ($\permutation{j}=\permutation{i}+1$).
We say that the \emph{boundary} point $p$ \emph{for} $x_i$ and $x_j$ is \emph{revealed}. 
The adversary must guarantee that at any time, for all existing atoms, at least one boundary is revealed.
The game ends as soon as an atom representing~$[0,\inputSize]$ exists.
\end{definition}

Note the following: 
The algorithm may try to fuse two atoms, even though by the revealed boundaries this is not guaranteed to succeed.
If this attempt is successful because their intervals share a point, a new atom is created and the algorithm solved a boundary.
We call this phenomenon a \emph{chance encounter}.
Since the \fusePEM extends the \movePEM, copying of atoms is allowed.

For the lower bound, we assume that the algorithm is omniscient. 
More precisely, we assume there exists a central processing unit, with unlimited computational power, that has all presently available information on the location of atoms and what is known about boundaries. 
This unit can then decide on how atoms are moved and fused. 

Thus, as soon as all boundary information is known to the algorithm, the instance is solvable with \bO{\log \inputSize} parallel I/Os:
The central unit can virtually list rank the atoms, group the atoms by rank into $P$ groups, and then by permuting move to every processor \bO{\memorySize} atoms which then can be fused with \bO{1} I/Os to the solving atom.

\createvar{binaryTree}{\mathcal{T}}{}
\createvar{binaryTreeOf}{\binaryTree}{i}

Hence, the careful revealing of the boundary information is crucial. 
To define the revealing process for GIF instances, the atoms and boundaries of a GIF instance \gifInstance of size \inputSize are related to a perfect binary tree \binaryTreeOf{\gifInstance}.
The tree \binaryTreeOf{\gifInstance} has~\inputSize leaves, $\inputSize -1$ internal nodes and every leaf is at distance $h=\log \inputSize$ from the root. 
More precisely, each leaf $i\in[\inputSize]$ corresponds to the atom representing the interval $[i-1,i]$.
And each internal vertex $v_p$ corresponds to the boundary~\mbox{$p=(i,j)$} where $i$ corresponds to the rightmost leaf of its left subtree, and $j$ to the leftmost leaf of the right subtree.
The levels of \binaryTreeOf{\gifInstance} are numbered bottom up: the leaves have level $1$ and the root vertex level $\log\inputSize$ (corresponding to the revealing order of boundaries).

The protocol for boundary announcement, shows for a random GIF instance~\gifInstance, that a deterministic algorithm takes \bOm{\log^2\inputSize} parallel I/Os.

\begin{definition}
\label{definitionAnnouncingGuide}
	The tree \binaryTreeOf{\gifInstance} is the \emph{guide} of \gifInstance if boundaries are revealed in the following way.
        Let $x$ be an atom of \gifInstance representing the interval~$I=[a,b]$.
        If neither of the boundaries  $a$ and $b$ are revealed, the boundary whose node in \binaryTreeOf{\gifInstance} has smaller level, is revealed. 
        If both have the same level, $a$ is revealed. 
\end{definition}

Note that for the analysis it is irrelevant how to break ties (in situations when the two invisible boundaries have the same level).
By the assumption that the algorithm is omniscient, when $p$ is revealed, immediately all intervals having $p$ as boundary know that they share this boundary.
Thus, the guide ensures that at any time each atom knows at least one initial atom with which it can be fused.

A node $v\in\binaryTreeOf{\gifInstance}$ is called \emph{solved}, if there is an atom of \gifInstance representing an interval that contains the intervals of the leaves of the subtree of~$v$.
The boundary $p$ is only revealed by the guide if at least one child of $v_p$ is solved. 

An easy (omniscient) algorithm solving a GIF instance can be implemented: 
in each of $O(\log \inputSize)$ rounds, permute the atoms such that for every atom there is at least one atom, known to be fuseable, which resides in the same cache. 
Fuse all pairs of neighboring atoms, reducing the number of atoms by a factor of at least~2, revealing new boundaries.
Repeat permuting and fusing until the instance is solved.
Because the permuting step can be achieved with~\bO{\log \inputSize} parallel I/Os, this algorithm finishes in~\bO{\log^2 \inputSize} parallel I/Os.
Note that solving a boundary resembles bridging out one element of an independent set in the classical list ranking scheme.  
Thus, most list ranking algorithms use information as presented to an algorithm solving a GIF instance \gifInstance guided by \binaryTreeOf{\gifInstance}.
Hence, this natural way of solving a GIF instance can be understood as solving in every of the $\log\inputSize$ rounds a proximate neighbors instance, making the \bOm{\log^2\inputSize} lower bound reasonable.

\createvar{levelOfTracedAtoms}{k}{}
\createvar{epsStages}{\varepsilon_s}{}
\createvar{numberStages}{\frac{2\log\memorySize}{16}}{}
\createvar{varStages}{s}{}

In the following we prove the lower bound for GIF by choosing~\levelOfTracedAtoms and showing that ~$\varStages=\levelOfTracedAtoms-1=\bO{\log \inputSize}=\bO{\log\memorySize}$ \emph{progress stages} are necessary to solve all nodes $W$ of level \levelOfTracedAtoms of \binaryTreeOf{\gifInstance} (thus, $|W|=2^{h+1-\levelOfTracedAtoms}$). 
For each stage we show in \autoref{lemmaGIFProgress} that it takes \bOm{\log\inputSize} parallel I/Os to compute.

\createvar{configuration}{C}{}
\createvar{configurationAt}{\configuration}{e}
\createvar{memoryOf}{M}{i}
\createvar{blockAt}{B}{i}
\createvar{memoryMultiplier}{c_m}{}
\createvar{configurationOfAt}{\configuration}{ie}
\createvar{configurationOf}{\configuration}{i} 

To measure the progress of a stage, configurations of \fusePEM machines are used.
The \emph{configuration} \configurationAt{t} after the \fusePEM machine performed $t$ I/Os followed by fusing operations consists of sets of atoms.
For each cache and each block of the shared memory, there is one set of atoms.

\createvar{boundaries}{\mathcal{B}}{}
\createvar{boundariesOf}{\boundaries}{i}

For~$e\in W$, let $T_e$ be the subtree of $e$ in \binaryTreeOf{\gifInstance}, and \boundariesOf{e} be all boundaries in $T_e$. 
The progress measure towards solving~$e$ is the highest level of a solved boundary in~\boundariesOf{e}.
More precisely, $T_e$ is \emph{unsolved} on level~$i$, if all boundaries of level~$i$ of~\boundariesOf{e} are unsolved. 
Initially every $T_e$ is unsolved on level $2$. 
The solved level increases by one at a time if only revealed boundaries are solved, but chance encounters may increase it faster.


\createvar{minLevelSize}{2^{\frac{15}{16}h}}{}

The execution of a deterministic algorithm~$A$ defines the following \varStages progress stages:
Let $\varStages=\numberStages$ 
and $X=\frac{|W|}{\varStages}$.
In each stage~$1 < i \leq \varStages$, 
 at least $X$ elements increase their level to~$i$.
Over time, the number of elements that are unsolved on level~$i$ decreases, and we define $t_i$ to be the last time where in $\configurationAt{t_i}$ the number of elements of~$W$ that are unsolved on level~$i+1$ is at least $|W|-iX$.
Further, let~$W_i$ be the elements (at least $X$ of them) that in stage~$i$ get solved on level~$i$ or higher (in the time-frame from $t_{i-1}$ to $t_i+1$).
We choose~$\levelOfTracedAtoms-1 = \frac{h}{16}$ such that $X = \frac{2^{h+1-k}}{\varStages}\geq 2^\frac{15h}{16}/\varStages = M^\frac{15}{8}/\varStages > M^\frac{7}{4}$ because $\varStages = \numberStages < M^\frac{1}{8}$ for $M\in\mathbb{N}$.


In the beginning of stage~$i$, for each $v \in W_i$ the level of~$v$ is at most~$i-1$, and hence all level~$i$ nodes are not announced to the algorithm. 
Let $P_i$ be the set of boundaries for which progress is traced:
For every $e\in W_i$, there is a node~$v_{p_e}$ of level~$i$ with boundary~$p_e$ that is solved first (brake ties arbitrarily). 
Then $P_i$ consists of those boundaries.
We define~$a_e$ and $b_e$ to be the two level~1 atoms (original intervals) defining the boundary~$p_e$.
Then all intervals having boundary~$p_e$ are derived of $a_e$ or~$b_e$.
Solving the boundary~$p_e$ means fusing any interval derived of $a_e$ with any interval derived of $b_e$. 
Furthermore a traced boundary is considered solved if in its interval (the one corresponding to an element of~$W$) a chance encounter solves a boundary of level greater than~$i$\todo{too edgy}.

To trace the progress of the algorithm towards fusing the atoms of one stage, we define the graph~$G_t^i=(V,E_t^i)$ from the configuration~\configurationAt{t_i+t}.
There is one vertex for each cache and each block of the shared memory (independent of~$t$).
There is an edge (self-loops allowed) $\{u,v\} \in E_t^i$ if for some~$e\in W_i$ some atom derived of~$a_e$ is at~$u$ and some atom derived of~$b_e$ is at~$v$ or vice versa.
The \emph{multiplicity} of an edge counts the number of such~$e$.
The multiplicity of the graph is the maximal multiplicity of an edge.

Note that solving a node~$v_e$ requires that it counts as a self-loop somewhere.
Hence the sum of the multiplicities of self-loops are an upper bound on the number of solved nodes, and for the stage to end, i.e., at time $t_{i+1}+1$, the sum of the multiplicities of loops must be at least~$X$.
After each parallel I/O chance encounters may happen.
Thus, the number of chance encounters is given by~$P$ times the  multiplicity of self-loops at the beginning of the stage.


We say that two nodes of \binaryTreeOf{\gifInstance} are indistinguishable if they are on the same level and exchanging them could still be consistent with the information given so far.
Let~$l_e$ (derived of~$a_e$) and $r_e$ (derived of~$b_e$) be the two children of~$v_e \in P_i$.
By definition, at time~$t_i$ both $l_e$ and~$r_e$ are unsolved and hence~$v_e$ is not revealed. 

Boundaries corresponding to nodes of level higher than \levelOfTracedAtoms may be announced or solved (not only due to chance encounters). 
To account for that, we assume that all such boundaries between the intervals corresponding to~$W$ are solved.
Hence the algorithm is aware of the leftmost and rightmost solved interval belonging to these boundaries, and this may extend to other intervals by revealed boundaries.
Only the nodes of level~$i$ that correspond to this leftmost or rightmost interval might be identifiable to the algorithm, all other nodes of level~$i$ are indistinguishable. 
Because $i<\levelOfTracedAtoms$, for all traced pairs $a_e,b_e$ at least one of the elements belongs to this big set of indistinguishable nodes.
We mark identifiable nodes. 
Hence, at stage $i$ the algorithm has to solve a random matching of the traced pairs where all marked nodes are matched with unmarked ones (and unmarked ones might be matched with marked or unmarked ones). 

\createvar{nrPorts}{|W_i|}{}
\createvar{locSize}{M}{}
\createvar{multiplicity}{m}{i}
\createvar{epsMultiplicity}{\varepsilon_m}{}

The next lemma derives a high-probability upper bound on the multiplicity of a graph. 
The full proof using the Hoeffding inequality is presented in \autoref{appendixGIF}.

\begin{lemma}\label{lemmaMultiplicity}
	Consider a deterministic GIF algorithm operating on a uniformly chosen permutation~$\pi$ defining \binaryTreeOf{\gifInstance} for the GIF instance \gifInstance. 
	Let $p(i,\memorySize)$ be the probability that $G^i_0$ has multiplicity at most~$\memorySize^{\frac{5}{8}}$ (where $P,\inputSize,t_j$, and \levelOfTracedAtoms  depend on parameter~\memorySize).
	Then there is a~$\memorySize^\prime$ such that for all $\memorySize\geq \memorySize^\prime$ and for each~$i \leq \levelOfTracedAtoms$ it holds $p(i,\memorySize) \geq  1-\frac{1}{\memorySize^2}$.
\end{lemma}

Fundamental insights on identifying two pairs of a $K_4$ show that the progress achieved with one I/O can not be too large.
A full proof is given in \autoref{appendixGIF}.

\begin{lemma}\label{lemmaQuadruple}
  If the graph $G^i_t$ has multiplicity at most~$m$, then $G^i_{t+1}$ has multiplicity at most~$4m$.
\end{lemma}

By the two previous lemmas we obtain the following result, which is proven in \autoref{appendixGIF}.

\begin{lemma}
	\label{lemmaGIFProgress}
  Let~$A$ be an algorithm solving an GIF instance \gifInstance guided by \binaryTreeOf{\gifInstance} of height~$h$ traced at level~$\levelOfTracedAtoms -1=h/16$.
  Each progress stage $j<\varStages=\levelOfTracedAtoms-1$ of $A$, assuming $t_j<\log^2 M$, takes time $t_j-t_{j-1}$=\bOm{\log M}.
\end{lemma}

There are \bO{\log\inputSize} stages, each taking at least \bOm{\log\inputSize} I/Os, yielding with a union bound over all $G^i_0$ for all progress stages $i<s$:

\begin{lemma}
	Consider a deterministic GIF algorithm operating on a uniformly chosen permutation~$\pi$ defining \binaryTreeOf{\gifInstance} for the GIF instance \gifInstance. 
	Then there is a~$\memorySize^\prime$ such that for all $\memorySize\geq \memorySize^\prime$, solving \gifInstance in the \fusePEM takes with high probability ($p>1-1/\memorySize$) at least \bOm{\log^2 \inputSize} parallel I/Os.
\end{lemma}

\createvar{alreadySolved}{U}{i}
\createvar{toBeSolved}{T}{i}
\createvar{notToBeSolved}{S}{i}

By Yao's principle \cite{1977YaosPrinciple} this can be transferred to randomized algorithms: 

\begin{theorem}
  The expected number of parallel I/Os to solve a GIF instance of size~$\inputSize=P\memorySize$ on an \fusePEM with $\memorySize=P$ is \bOm{\log^2 \inputSize}.
\end{theorem}

A simple reduction, given in \autoref{appendixGIF}, yields:

\begin{theorem}
\label{theoremBSPMapReduceGIF}
Solving the GIF problem in the BSP or in the MapReduce model with $\inputSize = P\memorySize$ and $P = \Theta(M) = \Theta(\sqrt{\inputSize})$ takes \bOm{\log \inputSize} communication rounds.
\end{theorem}


GIF is an attempt to formulate how the known algorithms for list ranking distribute information by attaching it to atoms of the PEM model.
Most known algorithms for list ranking use fusing of edges on an independent set of edge-pairs (bridging out edges). 
This means that every edge is used (if at all) either as first or second edge in the fusing. 
This choice of the algorithm is taken without complete information, and hence we take it as reasonable to replace it by an adversarial choice, leading to the definition of the guide of a GIF instance.

Additionally, the PEM lower bound shows that there is no efficient possibility to perform different stages (matchings) in parallel, showing that (unlike in the efficient PRAM sorting algorithms) no pipelining seems possible.
At this stage, our lower bound is hence more a bound on a class of algorithms, and it remains a challenge to formulate precisely what this class is. 
Additionally, it would be nice to show lower bounds in a less restrictive setting. 


\section{Upper Bounds}
\label{sectionUpperBounds}

Improvements in the analysis~\cite{2012ThesisGero} of the PEM merge sort algorithm~\cite{2008ArgeEtAlFundamentalPEM} yield:

\begin{lemma}[\cite{2012ThesisGero}]
\label{theoremMergeSortPEM}
The I/O complexity of sorting \inputSize records with the PEM merge sort algorithm using $P\leq\frac{\inputSize}{\blockSize}$ processors is $\sortp{\inputSize,\memorySize,\blockSize}=\bO{\frac{\inputSize}{P\blockSize}\lgb_d\frac{\inputSize}{\blockSize}}$ for $d=\max\{2,\min\{\frac{\inputSize}{P\blockSize},\frac{\memorySize}{\blockSize}\}\}$. 
\end{lemma}

We use it to extend in \autoref{appendixUpperBounds} the parameter range for the randomized list ranking algorithm~\cite{2010ArgeEtAlGraphAlgoPEM} from $P\leq\frac{\inputSize}{\blockSize^2}$, and $M=B^{\bO{1}}$ to:

\begin{theorem}
\label{theoremRuntimePEMListRanking}
The expected number of parallel I/Os, needed to solve the list ranking problem of size \inputSize in the CREW PEM model with $P\leq \frac{\inputSize}{\blockSize}$ is 
\begin{eqnarray*}
\mathcal{O}\left(\sortp{\inputSize,\memorySize,\blockSize}+(\log P)\lgb\frac{\blockSize}{\log P}\right) 
\end{eqnarray*}
which is for $\blockSize < \log P$ just \sortp{\inputSize,\memorySize,\blockSize}. 
\end{theorem}

In order to make the standard recursive scheme work, two algorithms are used.
By \autoref{theoremMergeSortPEM}, the randomized algorithm of \cite{2010ArgeEtAlGraphAlgoPEM}, which is based on \cite{1984VishkinRandomizedLR}, can be used whenever $\inputSize \geq P\min\{\log P,\blockSize\}$.
This yields the $(\log P)\lgb\frac{\blockSize}{\log P}$ term, if $\blockSize >\log P$.
Otherwise ($\inputSize\leq P\min\{\log P,B\}$) a simulation of a work-optimal PRAM algorithm \cite{1988AndersonMillerDetLR} is used.
A careful implementation and analysis yield \autoref{theoremRuntimePEMListRanking}.

\bibliographystyle{plain}
\bibliography{listrankingbib}

%
%
%
%

\newpage
\appendix

\section{Appendix}

\subsection{Number of Block Permutations of Proximate Neighbors}
\label{appendixProximateNeighbors}

\begin{lemma}
Any algorithm solving all proximate neighbors instances of size \inputSize must be capable of producing at least $\frac{(\inputSize/2)!}{(\blockSize/2)^{\frac{\inputSize}{2}}}$ different block permutations.
\end{lemma}

\begin{proof}

We restrict the inputs to a special class of the proximate neighbors problem:
The first $1\leq i\leq\frac{\inputSize}{2}$ atoms are labeled such that $\labelfunction{x_i}=i$ holds.
Hence there are $\frac{\inputSize}{2}!$ different instances of this type, as this is the number of matchings between the first and the second half of the input.

Note that some of the (seemingly) different inputs are solved by the same output block permutation.
We upper bound the number of different input instances one block permutation might solve.
Considering a block in an output block permutation, there are at most $x\leq\frac{\blockSize}{2}$ atoms of the first half of the input and $x$ atoms of the second half of the input.
The number of ways to understand this block as part of a proximate neighbors output is $x! \leq x^x$, since this is the number of ways to match the two types of atoms. 
In total, any output block permutation solves at most $(\blockSize/2)^{\frac{\inputSize}{2}}$ different input instances since there are $\frac{\inputSize}{2}$ pairs and all are in a block with at most $\frac{B}{2}$ pairs.
\qed
\end{proof}

\subsection{Proof of \autoref{theoremLowerBoundSemiGroup}}
\label{appendixSemigroup}

For proving the the correctness of the reduction in \autoref{theoremLowerBoundSemiGroup} we prove in \autoref{lemmaEverySemigroupOperation} the following fact on programs which use only the semigroup operation to solve an instance of the semigroup evaluation problem.
We show, a program which evaluates $\prod_{j=1}^{\inputSize}a_j$ has to compute for every $1\leq i \leq \inputSize-1$ intermediate results $x=\prod_{j=c}^{i}a_j\semigroupOperation\prod_{m=i+1}^{k}a_m$ with $1\leq c\leq i \leq k -1 \leq \inputSize-1$. 

To this end, we describe the semigroup operations performed by programs through graphs which are similar to algebraic circuits.

\createvar{intermediateResult}{b}{}
\createvar{calculationDAG}{D}{}
\createvar{calculationDAGof}{\calculationDAG}{i}
\createvar{program}{P}{}
\createvar{resultVariable}{r}{}

\begin{definition}
Let \program be a program to compute the product $\prod_{i=1}^{\inputSize}a_i$ of a semigroup.
Its calculation DAG \calculationDAGof{\program} is defined as follows: 
For each input value, intermediate result and the output variable, named result, there is a vertex in~\calculationDAGof{\program}.
For each application of the semigroup operation $z=x\semigroupOperation y$ where the operands are represented in~\calculationDAGof{\program} by $v_x$, $v_y$, and $v_z$, there are the directed edges $(v_x,v_z)$ and  $(v_y,v_z)$ in~\calculationDAGof{\program}.
We say that~$x$ is \emph{used} for~$y$ if there is a directed path from~$v_x$ to~$v_y$.
A calculation DAG is called normalized if all variables, corresponding to vertices, are used for the output vertex.

A calculation DAG \calculationDAG is called correct if, independently of the semigroup and the input values used, the output variable \resultVariable, corresponding to~$v_{\resultVariable}$, has as value $\prod_{i=1}^{\inputSize}a_i$.
\end{definition}

\begin{lemma}
If there is a correct calculation graph \calculationDAGof{\program}, then there is a normalized correct calculation graph \calculationDAG.
\end{lemma}
\begin{proof}
Removing vertices that are not used for the result from~\calculationDAGof{\program} cannot change its correctness. 
\qed
\end{proof}

\begin{lemma}
\label{lemmaAllProd}
Every intermediate result~\intermediateResult in a normalized correct calculation DAG~\calculationDAG for semigroup evaluation $\resultVariable=\prod_{i=1}^{\inputSize}a_i$ represents a product $\intermediateResult=\prod_{i=j}^{k}a_i$, with $1\leq j \leq k \leq \inputSize$.
\end{lemma}

\begin{proof}
Assume there is an intermediate result~\intermediateResult not of the claimed form.
Then~\intermediateResult can be written as $\intermediateResult=s\semigroupOperation{}a_j\semigroupOperation a_k\semigroupOperation t$, with $j+1 \neq k$ and where $s$ and $t$ are arbitrary products over the inputs.
By induction, every intermediate result~$u$ that uses~\intermediateResult can be written as $u = s^{\prime}\semigroupOperation a_j\semigroupOperation a_k\semigroupOperation t^{\prime}$, where $s^{\prime}$ and $t^{\prime}$ are arbitrary products over the inputs.
Since~\calculationDAG is normalized, also the result can be written in this form, which contradicts the correctness of the program:
Consider for example the concatenation semigroup, $a_j=\hbox{'a'}$, $a_{j+1}=\hbox{'b'}$, $a_k=\hbox{'c'}$ and all other~$a_i=\varepsilon$.
Then the correct result is 'abc', whereas the computed result contains the substring~'ac'.
\qed
\end{proof}

\begin{lemma}
\label{lemmaEverySemigroupOperation}
Given a normalized correct calculation DAG of a semigroup evaluation yielding an (intermediate) result $\resultVariable=\prod_{j=c}^{d}a_{j}$, with $c+1\leq d$, there is for every $i\in [c,d-1]$ a vertex~$v_t$ which applies the semigroup operation to intermediate results such that $x=\prod_{j=h}^{i}a_j\semigroupOperation{}\prod_{m=i+1}^{k}a_m$ with $1\leq c\leq h\leq i \leq k -1 \leq d-1\leq \inputSize -1$.
\end{lemma}
\begin{proof}
Strong Induction over $d-c$:

Base case $d-c=1$: $r=\prod_{j=c}^{d}a_{j} = a_{c}\semigroupOperation{}a_{d}= a_c \semigroupOperation{} a_{c+1}$ . \checkmark

By \autoref{lemmaAllProd}, the computation represented by $v_t$ is $\prod_{j=c}^{d}a_{j}=\prod_{m=c}^{\ell}a_{m} \cdot \prod_{o=\ell+1}^{d}a_{o}$.

Induction Step: 
If $i=\ell$: \checkmark.

If $c\leq i < \ell$: By induction the vertex can be found in the intermediate results, calculating $\prod_{j=c}^{\ell}a_{j}$. \checkmark

If $\ell +1 \leq i\leq d-1$: By induction the vertex can be found in the intermediate results, calculating $\prod_{j=\ell+1}^{d}a_{j}$. \checkmark
\qed
\end{proof}

For the following theorem, we assume that all PEM algorithms are {\em normalized}, that is, during each I/O a processor can be either active or idle and all active processors perform the same type of task (read or write). 
Note that the I/O complexity of a normalized PEM algorithm is at most twice as large as the I/O complexity of an algorithm that implements both, read and write, operations by different processors within the same I/O operation. 

For completeness we restate \autoref{theoremLowerBoundSemiGroup}.

\begin{theorem}
\label{theoremLowerBoundSemiGroupRestate} 
At least one instance of the semigroup evaluation problem of size~\inputSize requires \bOm{\permp{\inputSize,\memorySize,\blockSize}} parallel I/Os in the CREW \semigroupPEM model with $P\leq\frac{\inputSize}{\blockSize}$ processors.
\end{theorem}

\begin{proof}[of \autoref{theoremLowerBoundSemiGroup}]
We give a non-uniform reduction from the proximate neighbors problem to the semigroup evaluation problem.
Let~\instPN be an instance of the proximate neighbors problem of size~\inputSize over the input atoms \mbox{$\inputAtoms=\{x_i|i\in[\inputSize]\}$} with its labeling function~\labf.
In the following we consider the semigroup evaluation problem over the semigroup on the set $\inputAtoms^2$ with the semigroup operation $(a,b)\semigroupOperation(c,d) = (a,d)$, where $a,b,c,d\in \inputAtoms$.
The instance~\instSE of size~\inputSize of this semigroup evaluation problem is defined over the input atoms $a_i=(x_i,x_i)$, with~\mbox{$1\leq i\leq \inputSize$}, which is related by its permutation to \instPN: 
The permutation \perm{} of~\instSE is one of the permutations such that for all $i\in\left[\frac{\inputSize}{2}\right]$, $\{\perm{2i-1},\perm{2i}\} = \labf{}^{-1}(i)$ holds.

Let~\progSE be a program solving~\instSE with \tSE{} parallel I/Os for parameters \MSE and \BSE.
In the following we prove that \progSE, using $P$ processors, can be transformed into a program~\progPN  using the same number of processors, and solving~\instPN with $\tPN \leq 5\tSE$ parallel I/Os, where~\mbox{$\MPN=2\MSE$}, and $\BPN=2\BSE$.
Since for \tPN the lower bound is known by \autoref{lemmaProximateNeighborsLbPEM}, a lower bound follows for \tSE.

The key idea is to write for each application of the semigroup \linebreak oper\-ation~\mbox{$(a,b)\semigroupOperation(c,d)$} the pair $\{b,c\}$ as a result for \instPN to the output. 
By the observation on semigroup evaluation programs, it follows that, among others, all pairs $\{a_{\perm{2i-1}},a_{\perm{2i}}\} = \{x,y\}$, i.e., all atoms of \instPN with $\labf(x) = \labf(y)$, are written as output.
Note that this may affect the number of parallel I/Os heavily, as in every parallel I/O, there may be a processor writing up to $2\MSE$ pairs due to repeated application of the semigroup operation.
Therefore we argue that the program can be changed such that for each parallel I/O of~\progSE there are at most two output operations needed to write all pairs to the output, yielding~\mbox{$3\tSE$} parallel I/Os in total.

As mentioned before the theorem we assume that \progSE solves \instSE \linebreak with~\mbox{\tSE} parallel I/Os such that all processors perform either a write operation or a read operation.
From this we construct the program~\progPN{} that solves~\instPN with $\tPN{}=3\tSE{} + \frac{3\inputSize}{P\BPN}$ parallel I/Os. 
To this end we assume that \progSE is normalized if it fulfills the following invariant: 
If an input operation $s$ yields that intermediate results $u$ and $v$ are in the memory of one processor and $w=u\semigroupOperation v$ is contained in the result, then $w$ is calculated immediately after $s$.
Conversely, there is no read operation $s$ such that intermediate results $u$ and $v$ are in the memory of one processor then it is not possible that the semigroup operation is not applied if $w=u\semigroupOperation v$ is contained in the result.
Thus, for each intermediate result $v$ in an input block, there can be at most two intermediate results $u$ and $w$ such that $u\semigroupOperation v\semigroupOperation w$ has to be computed. 
Since each processor can read in one I/Os at most \BSE atoms there are at most $2\BSE$ pairs which are written to the shared memory. 

Assume there is an input operation in~\progSE not satisfying the invariant. 
Then it can be normalized by applying the semigroup operation $w=u\semigroupOperation v$ after the input operation $s$.
Further semigroup operations using $u$ are replaced by using~$w$, while intermediate results $x=v\semigroupOperation a_k$ are replaced by $x=a_k$.
This yields by the associativity of the semigroup still the correct intermediate results. 
Furthermore this normalization does not increase memory consumption of \progSE:
Since the storage, used temporarily for $v$ is not used anymore, and $u$ is replaced by $w$. 

Because~\progPN is expected to work on full blocks of size \BPN, but \progSE operates on blocks with \BSE semigroup atoms, an additional scan is needed to create blocks of the correct size.
This takes one block $x_1,\ldots,x_{\BPN}$ and writes two blocks $(x_1,x_1), (x_2,x_2),\ldots,(x_{\BSE},x_{\BSE})$ and $(x_{\BSE+1},x_{\BSE+1}),\ldots,(x_{\BPN},x_{\BPN})$. 
It needs~$\frac{3\inputSize}{P\BPN}$ parallel I/Os.

Since $\permp{\inputSize,2\MSE,2\BSE} \geq \frac{1}{4}\permp{\inputSize,\MSE,\BSE}$ holds, we obtain:
\begin{multline*}
 3\tSE{} + \frac{3\inputSize}{P\BPN} = \tPN \geq \\ \permp{\inputSize,2\MSE,2\BSE} \geq \frac{1}{4}\permp{\inputSize,\MSE,\BSE}
\end{multline*}

Because every correct program for~\tSE must read the complete input, we have $\tSE{}\geq \frac{\inputSize}{P\BSE}=\frac{2\inputSize}{P\BPN}$, such that the above inequality leads to 
\[
 \tSE{}  \geq  \frac{1}{18}\permp{\inputSize,\MSE,\BSE} . 
\]
\qed\end{proof}

\subsection{Proof of \autoref{lemmaAtomicEdgeContraction}}
\label{appendixAtomicEdgeContraction}

For completeness we restate \autoref{lemmaAtomicEdgeContraction}

\begin{theorem} 
\label{lemmaAtomicEdgeContractionRestate}
At least one instance of the atomic edge contraction problem of size \inputSize requires \bOm{\permp{\inputSize,\memorySize,\blockSize}} parallel I/Os in the CREW \linebreak\bridgingPEM model with $P\leq\frac{\inputSize}{\blockSize}$ processors.
\end{theorem}


\begin{proof}
Let \instSE be an instance of the semigroup evaluation problem of size \inputSize and \perm{} its permutation. 
Let \instanceAtomicListRanking be one instance of the atomic edge contraction problem of size \inputSize such that $e_{\perm{i}}=(\perm{i},\perm{i+1})$, where $1\leq i \leq \inputSize$, and the atom corresponding to $e_i$ is located in the memory cell $i$ in the shared memory (assume that $\perm{\inputSize+1} = \inputSize+1$). 

Let \programAtomicListRanking be a program solving \instanceAtomicListRanking.
Then this is transformed into a program~\progSE by the following steps.

We observe the following relationship between the inputs of \instSE and \instanceAtomicListRanking:
At the $i$-th location of \instSE is the element $x_i=a_{\perm{}^{-1}(i)}=a_j$ located.
The target of the $i$-th edge is $\perm{\perm{}^{-1}(i)+1}$ at which position in \instSE the element $x_{\perm{\perm{}^{-1}(i)+1}}=a_{\perm{}^{-1}(i)+1}=a_{j+1}$ is located.

Therefore replacing in \programAtomicListRanking the operation ``move atoms of semigroup elements'' by ``move the corresponding atoms of edges'', and ``merging two edges $(a,b)$ and~$(b,c)$'' by of ``applying the semigroup operation'', yields program \progSE solving~\instSE.
By \autoref{theoremLowerBoundSemiGroup} we can assume that \instSE is an instance for which it takes at least \bOm{\permp{\inputSize,\memorySize,\blockSize}} parallel I/Os to solve. 
Thus, to solve \instanceAtomicListRanking, there are at least \bOm{\permp{\inputSize,\memorySize,\blockSize}} parallel I/Os needed.
\qed
\end{proof}

\subsection{Guided Interval Fusion}
\label{appendixGIF}

\subsubsection{Proof of \autoref{lemmaMultiplicity}}
For completeness, we restate \autoref{lemmaMultiplicity}.
\begin{lemma}\label{lemmaMultiplicityRestate}
	Consider the family of GIF algorithms depending on parameter~\memorySize (and thus defining $P,\inputSize,t_j$, and \levelOfTracedAtoms), and the graph~$G^i_0$ defined by the snapshot at time~$t_i$.
	Let $p(i,\memorySize)$ be the probability that $G^i_0$ has multiplicity at most~$\memorySize^{\frac{5}{8}}$ for uniformly chosen permutation~$\pi$ defining \binaryTreeOf{\gifInstance} for the GIF instance \gifInstance. 
	Then there is a~$\memorySize^\prime$ such that for all $\memorySize\geq \memorySize^\prime$ and for each~$i \leq \levelOfTracedAtoms$ it holds $p(i,\memorySize) \geq  1-\frac{1}{\memorySize^2}$.
%
\end{lemma}
\begin{proof}
The total number of traced vertices is~$2|W_i| > M^\frac74$.
Drawing the permutation of the leafs of the GIF instance uniformly at random implies that the nodes of level~$l$ of the tree from a uniform matching among the subtrees of their children.
The marked subtrees are identifiable, but the unmarked ones are indistinguishable, and the matching does not connect marked subtrees.
Hence, the random experiment is the following:
There are marbles representing the traced atoms.
At most half of the marbles are marked by unique numbers, the others are unmarked.
Consider two arbitrary vertices $u$ and $v$ of $G^i_0$ ($u=v$ is allowed). 
This is reflected by the algorithm choosing two arbitrary sets $M_u$ and~$M_v$ of marbles of size at most~\memorySize.
Drawing the edges of the matching (a uniform one because $\pi$ was uniform) can be done by considering the marbles in an arbitrary order and choosing the matching partner uniformly at random from the remaining possible neighbors.
We chose as order to first take all marked marbles at~$u$ (and choose a neighbor uniformly at random from the remaining unmarked marbles), and then the (remaining) unmarked marbles at~$u$ (choose a random marked marble). 
The random variable~$Y$ we are interested in is the number of times this random choice leads to a connection to~$M_v$.
In every step the conditional probability of increasing~$Y$ is at most $p'=\frac{M}{\frac{1}{2}(\nrPorts-2M)}$.
Hence, the expected value of~$Y$ is asymptotically at most~$\memorySize^2/\nrPorts < \memorySize^{\frac{1}{4}}$.
We define $Y$ to be big if $Y\geq \memorySize^{\frac{5}{8}}$.
To apply a Hoeffding inequality, we can upper bound~$Y$ by the following Bernoulli experiment:
We $n'=\memorySize$ times toss a coin with success probability~$p'$, and as for the probability to have more than $T=(p'+\varepsilon')n'$ successes with $T=\memorySize^{\frac{5}{8}}$ and hence $\varepsilon'=\memorySize^{\frac{-3}{8}}$ 
is good enough. 
Then the Hoeffding inequality states that $Y$ is big with probability at most $p=e^{-2\varepsilon'^2n'}$.
We calculate $\varepsilon'^2n'= (M^\frac{-3}{8}/2)^2M= M^{\frac{-6}{8}+1}/4=M^\frac{1}{4}/4$, leading to $p=e^{-M^\frac{1}{4}/2}$ as the probability of big~$Y$.
Interpreting this in the context of~$G^i_0$, the probability that any edge or self-loop has  high multiplicity is at most~$p$.
To bound the overall multiplicity of $G^i_0$ we use a union bound.
Because the lower bound is only useful for algorithms performing at most~$\log^2 \memorySize$ parallel I/Os, the number of vertices in~$G^i_0$ is at most $\frac{\inputSize}{\blockSize} + P\log^2 \inputSize= 2\memorySize + \memorySize 4\log^2 \memorySize$, and hence the number of edges (pairs) is at most $28 \memorySize^2\log^4 \memorySize$.
Multiplying this with $p$ we get that the probability that any edge in the graph has multiplicity more than $M^\frac58$ is at most $28 M^2\log^4 M e^{-M^\frac{1}{4}/2}$ which tends to zero quicker than $1/M^2$.
\qed
\end{proof}

\newpage
\subsubsection{Proof of \autoref{lemmaQuadruple}}
For completeness, we restate \autoref{lemmaQuadruple}.

\begin{lemma}\label{lemmaQuadrupleRestate}
  If the graph $G^i_t$ has multiplicity at most~$m$, then $G^i_{t+1}$ has multiplicity at most~$4m$.
\end{lemma}
\begin{proof}
  Observe that a derived atom of the traced elements do not appear out thin air, but they must be copied around.
  If the I/O operation between the two configurations is a write operation, the graph~$G^i_{t+1}$ is~$G^i_t$ with (some of) the nodes standing for processors split, distributing or copying edges.
  Obviously, this does not increase the multiplicity.
  If the I/O operation is a read operation, then each node representing a processor can be united with one representing a block.
  In the worst case the new edge considered is between two such united nodes, and the multiplicity can stem from four former edges, see Figure~\ref{figureNodeMerging}.
\createvar{vertexA}{v_{1}}{}
\createvar{vertexB}{v_{2}}{}
\createvar{vertexC}{v_{3}}{}
\createvar{vertexD}{v_{4}}{}
\createvar{matchingLeft}{m_1}{}
\createvar{matchingRight}{m_2}{}
	\begin{figure}
\centering
	\begin{tikzpicture}[every loop/.style={}]
	\node[node,label=below:\vertexA] (v_1) at (0,0) {} edge [in=100,out=200,loop] node[below left]  {$\ell_1$} ();
	\node[node,label=above:\vertexB] (v_2) at (0,1) {} edge [in=260,out=170,loop] node[above left]  {$\ell_2$} ();
	\node[node,label=below:\vertexC] (v_3) at (2,0) {} edge [in=80,out=340,loop] node[below right]  {$\ell_3$} ();
	\node[node,label=above:\vertexD] (v_4) at (2,1) {} edge [in=20,out=280,loop] node[above right]  {$\ell_4$} ();
	
	\draw (v_1) -- node[below] {$e_1$} (v_3);
	\draw (v_2) -- node[pos=0.85,above] {$e_2$} (v_3);
	\draw (v_1) -- node[pos=0.15,above] {$e_3$} (v_4);
	\draw (v_2) -- node[above] {$e_4$} (v_4);
	\draw[dashed] (v_1) -- node[left] {\matchingLeft}  (v_2);
	\draw[dashed] (v_3) -- node[right] {\matchingRight}(v_4);
	\end{tikzpicture}
	\caption{Contracting two edges of a matching in } 
	\label{figureNodeMerging}
	\end{figure}
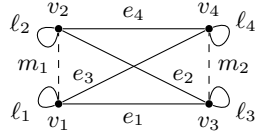
  If we consider a self loop, it can stem from two self-loops and one edge.
  Hence each edge of~$G^i_{t+1}$ has multiplicity at most four times the multiplicity of~$G^i_t$.
\qed
\end{proof}

\subsubsection{Proof of \autoref{lemmaGIFProgress}}
For completeness, we restate \autoref{lemmaGIFProgress}.

\begin{lemma}
  Let~$A$ be an algorithm solving an GIF instance \gifInstance guided by \binaryTreeOf{\gifInstance} of height~$h$ traced at level~$k=h/16$.
  Each progress stage $j<\varStages=\levelOfTracedAtoms-1$ of $A$, assuming $t_j<\log^2 M$, takes time $t_j-t_{j-1}$=\bOm{\log M}.
\end{lemma}
\begin{proof}
  At time $t_j$, the multiplicity of $G^i_{t_j}$ is by Lemma~\ref{lemmaMultiplicity} at most~$M^\frac58$, and by Lemma~\ref{lemmaQuadruple} at time $t_j+t$ at most $4^tM^\frac58$.
  By the assumption $t_j<\log^2 M$,  the number of vertices in~$G^i_{t'}$ is at most  $P\log^2M$.
  Hence, at time $t'=t_j+t$, the total multiplicity~$T$ of loops in $G^i_{t'}$ is at most $M^\frac58  P 4^t \log^2M$, and the number of solved traced pairs including chance encounters is at most $M^\frac58  P (4^t + t) \log^2M$
  For the stage to finish, this number must be at least~$X>M^\frac{7}{4}$, 
  Hence we have $M^\frac58  M (4^t+t) \log^2M>M^\frac{7}{4}$, implying $4^t+t > M^{\frac74-1-\frac58}/\log^2M = M^\frac18/\log^2M$ and hence $t=\bOm{\log M}$.
\qed
\end{proof}

\subsubsection{Proof of \autoref{theoremBSPMapReduceGIF}}
For completeness, we restate \autoref{theoremBSPMapReduceGIF}.

\begin{theorem}
\label{theoremBSPMapReduceGIFRestate}
Solving the GIF problem in the BSP or in the MapReduce model with $\inputSize = P\memorySize$ and $P = \Theta(M) = \Theta(\sqrt{\inputSize})$ takes \bOm{\log \inputSize} communication rounds.
\end{theorem}
\begin{proof}
Communication of each round in the distributed models can be implemented in the PEM model in $O(\permp{\inputSize, \memorySize, \blockSize})$ parallel I/Os, which are in this parameter setting \bO{\log P} parallel I/Os. 
Therefore, $o(\log \inputSize)$ communication rounds for $P = \Theta(M) = \Theta(\sqrt{\inputSize})$, would imply $o(\log^2 \inputSize)$ parallel I/Os in the PEM model.  
\qed
\end{proof}

\subsection{Upper Bounds}
\label{appendixUpperBounds}

In the following we extend the parameter range of the randomized list ranking algorithm of \cite{2010ArgeEtAlGraphAlgoPEM} for the CREW PEM model. 
In a first part we present the algorithmic concepts \cite{2010ArgeEtAlGraphAlgoPEM,2012ThesisGero,1984VishkinRandomizedLR,1991AndersonMillerRanomizedWorkOptLR}.
In a second part we analyze the runtime for different parameter settings of \inputSize, \memorySize, \blockSize and $P$. 
Every linked list can be converted into a double linked list by permuting twice \cite{2010ArgeEtAlGraphAlgoPEM}. 
Since this matches the lower bound of \autoref{sectionListRankingLB} it is reasonable to assume that the input is a double linked list.

\begin{figure}
\centering
\begin{tikzpicture}
\def\leftLabel{-0.25}
\def\bottomLabel{-0.25}
\def\rightBorder{8}
\def\midBorder{4}
\def\pPos{1}
\def\plgPos{2}
\def\pbPos{3}
\def\bigPos{4}
\coordinate (0Left) at (0,0);
\coordinate (0Mid) at (\midBorder,0);
\coordinate (0Right) at (\rightBorder,0);
\coordinate (PLeft) at (0,\pPos);
\coordinate (PRight) at (\rightBorder,\pPos);
\coordinate (PlgLeft) at (0,\plgPos);
\coordinate (PlgMid) at (\midBorder,\plgPos);
\coordinate (PBLeft) at (0,\pbPos);
\coordinate (PBMid) at (\midBorder,\pbPos);
\coordinate (PBRight) at (\rightBorder,\pbPos);
\coordinate (BigLeft) at (0,\bigPos);
\coordinate (BigMid) at (\midBorder,\bigPos);
\coordinate (BigRight) at (\rightBorder,\bigPos);
\fill[blue!20!white] (0Left) -- (0Right) -- (PBRight) -- (PBMid)  -- (PlgLeft);
\shade[bottom color=green!40!white, top color=green!0!white] (PBLeft) -- (PBRight) -- (BigRight) -- (BigLeft) ;
\fill[red!40!white] (PlgLeft) -- (PBMid) -- (PBLeft);
\draw[black] (0Left) -- (0Right);
\draw[black,dashed] (PLeft) -- (PRight);
\draw[black] (PlgLeft) -- (PBMid);
\draw[black] (PBLeft) -- (PBRight);
\draw[black] (0Left) -- (PBLeft);
\draw[black] (0Right) -- (PBRight);
\draw[black,dashed,->] (PBLeft) -- (BigLeft); 
\draw[black,dashed] (PBMid) -- (BigMid);
\draw[black,dashed] (PBRight) -- (BigRight);
\draw[black,dashed] (0Mid) -- (PBMid);
\node[anchor=east] at (\leftLabel,0) {$N=1$};
\node[anchor=east] at (\leftLabel,\pPos) {$N=P$};
\node[anchor=east] at (\leftLabel,\plgPos) {$N=P\log P$};
\node[anchor=east] at (\leftLabel,\pbPos) {$N=PB$};
\node[anchor=north] at (\midBorder,\bottomLabel) {$B=\log P$};
\node[anchor=north] at (6.5,\bottomLabel) {$B<\log P$};
\node[anchor=north] at (1.5,\bottomLabel) {$B>\log P$};
\end{tikzpicture}
\caption{Illustration of the Parameter Range}
\label{figureListRankingParameters}
\end{figure}
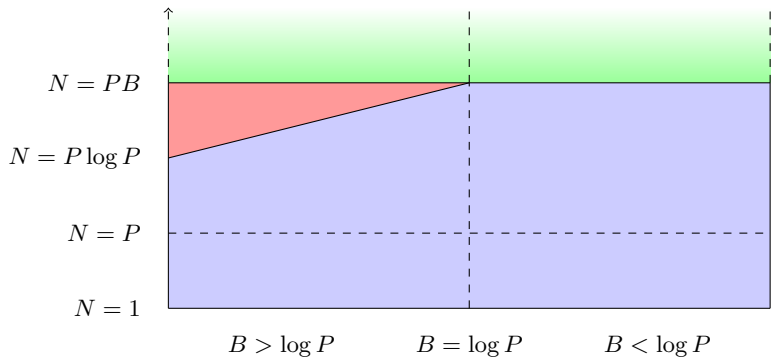
\createvar{smallInstanceSize}{\inputSize^\prime}{}

\subsubsection{Algorithms for Different Parameter Settings}

We use two algorithms for the different parameter ranges depicted in \autoref{figureListRankingParameters}.
Both use the concept of repeatedly bridging out an independent set $S$ of list elements, until the list is of constant size and therefore can easily be solved.
By reversing the process of bridging out elements the rank of all elements of the input list can be determined \cite{2010ArgeEtAlGraphAlgoPEM,1995ChiangExternalMemoryGraphAlgorithms,1991AndersonMillerRanomizedWorkOptLR}. 
Bridging out $S$ is technically done by setting in parallel for all $x\in S$ the link of the predecessor of $x$ to the successor of $x$. 
This can be done by a constant number of sorting operations \cite{2010ArgeEtAlGraphAlgoPEM}.

Therefore, we consider the slight differences of selecting an independent set:
\emph{Algorithm 1} uses the basic idea of Vishkin~\cite{1984VishkinRandomizedLR} to solve the list ranking problem.
It is proven to work efficiently for $P\leq \frac{\inputSize}{\blockSize^2}$ and $\memorySize = \blockSize^{\bO{1}}$ in \cite{2010ArgeEtAlGraphAlgoPEM}. 
The algorithm computes a random independent set $S$ of expected size $\frac{\inputSize-1}{4}$, by tossing a coin for each list element and selecting those elements for $S$ whose coin toss yields~$1$ and whose successors coin yields $0$.
Bridging out yields a smaller list which is processed recursively.
Algorithm 1 is used whenever $\inputSize \geq P\min\{\log P,\blockSize\}$ (green and red area in \autoref{figureListRankingParameters}).

Another well known, randomized algorithm for list ranking \smallInstanceSize elements in the PRAM model was presented by Anderson and Miller~\cite{1991AndersonMillerRanomizedWorkOptLR}.
We give a PEM version of it as Algorithm 2 and use it whenever  $\smallInstanceSize \leq P \min\{\blockSize,\log P\}$ (blue area in \autoref{figureListRankingParameters}):

\emph{Algorithm 2} assigns to at most $P$ processors a queue of $\min\{\log P,\blockSize\}$ list elements which are located in consecutive cells of the hard disk, which is by definition at most one block. 
Every processor bridges the head element of its queue out (this is called a round) until all queues are empty.
Thus there is no recursive processing.
If two processors try to bridge out successive list elements, the same random protocol to break ties, as used in Algorithm 1.
Anderson and Miller prove in \cite{1991AndersonMillerRanomizedWorkOptLR}, that the probability, that a queue is not empty after $16\log P$ rounds is $P^{-\frac{9}{8}}$. 

\subsubsection{Analysis for Different Parameter Settings}

A deeper analysis of the merge sort algorithm of \cite{2008ArgeEtAlFundamentalPEM} yields improvements in the parameter range (cf. \autoref{theoremMergeSortPEM}).


Therefore reducing the size of a list ranking instance with size $\inputSize \geq P\blockSize$ to an instance of size~$\smallInstanceSize \leq P\blockSize$ with Algorithm 1 takes at most \sortp{\inputSize,\memorySize,\blockSize}~I/Os: 
the number of parallel I/Os needed to reduce a list of size \inputSize to a list of size~\smallInstanceSize is \bO{\sortp{\inputSize,\memorySize,\blockSize}} due to two facts.
The list sizes of the recursive algorithm are geometrically decreasing and the first sorting steps are dominant since $\frac{3}{4}\sortp{\inputSize,\memorySize,\blockSize} \geq \sortp{\frac{3}{4}\inputSize,\memorySize,\blockSize}$ \cite{1995ChiangExternalMemoryGraphAlgorithms,2010ArgeEtAlGraphAlgoPEM}. 

We conclude the complexity of the list ranking problem by giving efficient algorithms for solving the last \smallInstanceSize list elements. 
If it is possible to show that this can be done in \bO{\log P} I/Os, this yields that list ranking takes \sortp{\inputSize,\memorySize,\blockSize}~I/Os, since \sortp{\inputSize,\memorySize,\blockSize} becomes \bO{\log P} when sorting $P\blockSize$ elements.
However this will not be possible in all cases.

If $\smallInstanceSize\leq P\log P$ we use Algorithm 2 to obtain a bound of \bO{\log P} I/Os to solve the list ranking problem of the remaining size. 
Assigning to each processor the at most $\min\{\log P,\blockSize\}$ list elements can be done with two parallel I/Os, since the algorithm runs on a CREW PEM.

Note that in Algorithm 2 bridging out the independent set $S$ of list elements of a round can be done with constant number of parallel I/Os by using direct processor to processor communication since there is at most one element bridged out per processor.
Since a queue of Algorithm 2 is not empty after \bO{\log P} rounds with probability~$P^{-\frac{9}{8}}$, a union bound and some calculations show that the probability that there is a queue which is not empty after \bO{\log P} I/Os is~$\smallInstanceSize P^{-\frac{9}{8}} = \bO{\frac{1}{\log \smallInstanceSize}}$ (using $P\leq \smallInstanceSize \leq P\log P$).
Thus, with high probability, there are at most \bO{\log(P)} parallel I/Os needed to solve a list ranking instance of size \smallInstanceSize.

Thus there is one gap left (the red area in \autoref{figureListRankingParameters}), if $\blockSize > \log P$ and if~$P \log P \leq \inputSize^\prime \leq P\blockSize$.
In this case we use Algorithm 1 and thus the expected size of the instance in each round is reduced by at least a fourth.
This yields at most~$\bO{\log(P\blockSize) -\log(P\log P)}=\bO{\log\frac{B}{\log P}}$ rounds.
As mentioned, in this parameter setting, $\smallInstanceSize\leq P\blockSize$, sorting and therefore bridging out takes at most~\bO{\log P} parallel I/Os.
Therefore the number of parallel I/Os in this parameter range can be bounded by \bO{(\log\frac{B}{\log P})(\log P)}.

In total this leads to \autoref{theoremRuntimePEMListRanking}, which is restated for completeness:

\begin{theorem}
\label{theoremRuntimePEMListRankingRestated}
The expected number of parallel I/Os, needed to solve the list ranking problem of size \inputSize in the CREW PEM model with $P\leq \frac{\inputSize}{\blockSize}$ is 
\begin{eqnarray*}
\mathcal{O}\left(\sortp{\inputSize,\memorySize,\blockSize}+(\log P)\lgb\frac{\blockSize}{\log P}\right) 
\end{eqnarray*}
which is for $\blockSize < \log P$ just \sortp{\inputSize,\memorySize,\blockSize}. 
\end{theorem}

\end{document}